\documentclass[a4paper]{llncs}
\pdfoutput=1
\usepackage{amssymb}
\usepackage{graphicx}
\usepackage{url}
\urldef{\mailsa}\path|{alfred.hofmann, ursula.barth, ingrid.haas, frank.holzwarth,|
	\urldef{\mailsb}\path|anna.kramer, leonie.kunz, christine.reiss, nicole.sator,|
	\urldef{\mailsc}\path|erika.siebert-cole, peter.strasser, lncs}@springer.com|    
\newcommand{\keywords}[1]{\par\addvspace\baselineskip
	\noindent\keywordname\enspace\ignorespaces#1}

\begin{document}
	
	\mainmatter  
	
	\title{Competitive Analysis of Move-to-Front-or-Middle (MFM)\\ Online List Update Algorithm}
	
	\titlerunning{Lecture Notes in Computer Science: Authors' Instructions}
	
	%
	%
	\author{Baisakh, Rakesh Mohanty}
	%
	
	\institute{Veer Surendra Sai University of Technology,\\Department of Computer Science and Engineering,\\
		Burla, Sambalpur,Odisha, 768019}
	
	%
	%
	
	\maketitle

\begin{abstract}

The design and analysis of efficient algorithms with the knowledge of current and past inputs is a non trivial and challenging research area in computer science. In many practical applications the future inputs are not available to the algorithm at any instance of time. So the algorithm has to make decisions based on a sequence of inputs that are in order and on the fly. Such algorithms are known as online algorithms. For measuring the performance of online algorithms, a standard measure, known as competitive analysis, has been extensively used in the literature. List update problem is a well studied research  problem in the area of online algorithms since last few decades. One of the widely used deterministic online list update algorithm is the \emph{Move-To-Front (MTF)} algorithm, which has been shown to be $2-competitive$ with best performance in practical real life inputs. In this paper we analyse the \emph{Move-to-Front-or-Middle (MFM)} algorithm using competitive analysis by addressing one of an open question raised by \emph{Albers} that whether dynamic offline algorithm can be used in finding the competitiveness of an online algorithm? \emph{Move-To-Front-or-Middle (MFM)} was experimentally studied and observed to be performing better than \emph{MTF} algorithm by using the Calgary Corpus and Canterbury Corpus data set. However, it is interesting and challenging to find the lower bound and upper bound on the competitive ratio of \emph{MFM} algorithm. We make a first attempt to find the competitiveness of \emph{MFM} algorithm. Our new results show that \emph{MFM} is not \emph{2-competitive} with respect to static optimum offline algorithm, whereas it is \emph{2-competitive} with respect to dynamic optimum offline algorithm. Our new theoretical results may open up a new direction of research in the online list update problem by characterising the structure of competitive and non competitive deterministic online algorithms.

	\end{abstract}
	
\section{Introduction}
\subsection{List Update Problem} 
	The \emph{List Update Problem (LUP)} introduced by \emph{McCabe} \cite{McCabe} has been a well studied classical optimization problem in the literature \cite{Bentley}, \cite{Rivest} \cite{Binter} since last five decades. In this problem we have two inputs such as an unsorted linear list \emph{L} and a sequence of requests $\sigma$. The list \emph{L} contains a set of distinct \emph{l} items i.e $<x_{1} , x_{2} , ..., x_{l}>$. These set of items need to be maintained by the list update algorithm while serving a sequence of requests $\sigma=$ $ <\sigma_{1}, \sigma_{2}, ..., \sigma_{n} >$. Every single request $\sigma_{i}$ refers to an access operation to the requested item present in the list where  $\sigma_{i} \in L$. To access an item, the list update algorithm begins the search operation linearly from the front of the list. It keeps searching for the requested item towards the end of the list till the item is found. Every time an item is served, the list configuration can be changed by performing the list reorganization through exchange or reordering of items of  the list. A configuration of the list represents any permutation of $l$ items of the list. The list can be reorganized by considering two different types of exchange operations such as  free exchanges and the paid exchanges. The total access and reorganization cost of serving a request sequence on the list is determined by using a cost model. The standard cost model defined by \emph{Sleator and Tarjon} \cite{Sleator} has been extensively used in the literature for the list update problem.
	
	In the standard cost model, accessing an item $x_{i}$ at the $i^{th}$ position of the list incurs a cost of \emph{i}. Once an item is accessed it can be moved to any position forward of the list with no extra cost which is known as a free exchange. But if any two consecutive items are exchanged with respect to their relative positions, then it incurs a cost of \emph{1}, and is known as a paid exchange. The objective of the list update algorithm is to minimize the total cost denoted by $C(L, \sigma)$ which is the sum of access cost and reorganization cost while serving a request sequence $\sigma$.
	
\subsection{Online Algorithms and Competitive Analysis} In an online scenario, a list update algorithm serves every request without having the knowledge of future requests. Such an algorithm is known the online list update algorithm and is denoted by $ALG$. At every time instance $t_{i}$ and input $\sigma_{i}$ is known to the algorithm where as the rest of requested items i.e $< \sigma_{i+1}, \sigma_{i+2} , ..., \sigma_{n}>$ are unknown. 
	
	Though many deterministic online list update algorithms have been proposed in the literature \cite{Bachrach}, most of the algorithms are the variants of three well known primitive algorithms. They are \emph{Move-To-Front (MTF)}, \emph{Transpostion (TRANS)} and \emph{Frequency Count (FC)} as follows.\\
	\emph{\bf{Move-To-Front (MTF):}} When an item is accessed, then it is moved immediately to the front of the list with only free exchanges.\\  
	\emph{\bf{Transpose (TRANS):}} An immediately accessed item is moved to one position forward in the list by transposing with its immediately preceding item.\\  
	\emph{ \bf{Frequency Count (FC):}} When an item is accessed, the count value is incremented to one and then sort the items in the list according to the decreasing order of their frequency.\\  
	Among these three algorithms, \emph{MTF} is most widely used in practical applications and well studied in the literature.
	
	Competitive analysis \cite{Manasse} is one of the standard performance measures for the online algorithms \cite{Borodin1} \cite{Fiat}. Competitive analysis is performed by comparing the worst case cost incurred by an online algorithm \emph{ALG} over all finite request sequences, with the cost incurred by the optimum offline algorithm \emph{OPT}. Optimum offline algorithm gives the minimum cost while serving any given request sequence $\sigma$. Let the cost incurred by $ALG$ and $OPT$ be denoted by $C_{ALG}(L,\sigma)$ and $C_{OPT}(L,\sigma)$ respectively. Then an online algorithm \emph{ALG} is called \emph{d-competitive} if there exist a constant $\beta$ such that $C_{ALG}(L,\sigma)$ $\leq$ $d.C_{OPT}(L,\sigma)$ $+\beta$ for all request sequences of $\sigma$. The less the competitive ratio the better the performance of an online algorithm.
	
\subsection{Practical and Research Motivations:} One of the variant of \emph{Move-To-Front (MTF)}, known as \emph{Move-to-Front-or-Middle(MFM)} was proposed  and experimentally studied in \cite{Mohanty et. al}. The nature of both the algorithms are almost similar to each other. \emph{MFM} differs from \emph{MTF} only when the position of the immediately accessed item is greater than the middle position of the list. In that case the \emph{MFM}, moves the accessed item to the middle position, instead of moving to the front of the list.
	
	The experimental study of \emph{MFM} was conducted on two well known benchmark data sets, \emph{Calgary Corpus} and \emph{Canterbury Corpus}. In experimental analysis, they compared the performance of the proposed \emph{MFM} and the \emph{MTF} in terms of total gain. The results showed that \emph{MFM} outperformed \emph{MTF}, though it has not strikingly different performance from the \emph{MTF} on the tested benchmarks. According to our knowledge, till date no theoretical analysis has been performed to study the competitive nature of the \emph{MFM} online algorithm. In this paper, we begin with the motivation to analyse whether \emph{MFM} is a \emph{2-competitive} algorithm or not? So the primary objective of our study is to investigate the underlying gap that exists between the experimental and  theoretical results of \emph{MFM}.
	
	Our theoretical study of \emph{MFM}, makes an attempt to address one an open question posed by \emph{Albers} in \cite{Albers}, whether we can develop competitive online algorithm with respect to dynamic offline algorithm or not? Because in general, the competitive analysis of an online algorithm is performed against the static offline algorithms. The static offline algorithms perform the list rearrangement only once at the beginning, prior to serve the given request sequence and the list remains unchanged till all the items are served. It is  denoted as \emph{STAT}, that performs the reordering of list according to the non increasing order of the request frequencies, and then does not change the list configuration. When an online algorithm denoted by \emph{ALG} is \emph{k-competitive} against the \emph{STAT}, then it can be formally presented as $C_{ALG} \leq k.C_{STAT}+ \beta$ for all request sequence $\sigma$, where $C_{ALG}$ and $C_{STA}$ are the total cost of serving the request sequence $\sigma$. In this paper, our central focus is to investigate whether $C_{MFM} \leq k.C_{STAT}+ \beta$ or not? Can we design a dynamic optimum offline algorithm, denoted by $DYN$, where we can claim that \emph{MFM} is competitive with respect to \emph{DYN} i.e $C_{MFM} \leq k.C_{DYN}+ \beta$. More ever, one of our results depicts that the static optimum offline algorithm is always associated with an important factor that decides whether a deterministic online algorithm is competitive or not. We show that the number of elements participate in the worst case analysis of an online deterministic algorithm, influences the competitive nature of the online algorithm.
	
\subsection{Our Line of Work}
Our line of work has been motivated by the practical application of list update problem in dynamic compression model \cite{Borodin1}. In dynamic compression model, the initial list does not contain any item initially. Upon encountering a word that needs to be compressed, it is first inserted in the list and so the list is configured dynamically. Hence our theoretical study on the competitive analysis is based on this dynamic model, where we assume that the cruel request sequence being picked up by the adversary contains all items of the list. But in the literature \cite{Borodin1}, we have observed that while performing competitive analysis, the worst case cruel request sequence for different algorithms are often not typical in real life. For an example of TRANS, the cruel sequence always requests to only the last two items of the list, which gives a total worst case cost of \emph{nl}, where \emph{l} is the size of the list and \emph{n} is the size of the request sequence. In the practical application such as compression, such request sequence does not have much significance. So to avoid this pessimistic nature of competitive analysis, we have considered the cruel request sequences which are more practical while performing the competitive analysis.
	
\section{Related Work}
\subsection{Well Known Results on Deterministic List Update Algorithms}
	An empirical study of \emph{LUP} performed by \emph{Bentley} and \emph{McGeoch} in \cite{Bentley}, showed that \emph{MTF} and \emph{FC} are \emph{2-competitive} with respect to static optimum offline algorithm. Whereas \emph{TRANS} was found to be non competitive. \emph{Sleator and Tarjan} proposed a theoretical amortized analysis on these three algorithms in one of their seminal paper in \cite{Sleator}, where only \emph{MTF} is found to be \emph{2-competitive} for a dynamic list, whereas \emph{TRANS} and \emph{FC} are not competitive. 
	
	\subsection{Move-to-Front-or-Middle Algorithm (MFM)}
	\emph{MFM}\cite{Mohanty} is a variant of \emph{MTF} algorithm in which an accessed item in the list is moved either to the front or to the middle of the list based on some condition. The algorithm is presented as follows.
	
	Let \emph{L} be a list of size \emph{l} and \emph{m} be the middle position of the list, where $m= \lceil{\frac{l}{2}}\rceil$. Let \emph{x} be the currently accessed item in the list at position \emph{p}, where $1 \leq p \leq l$. If $p \leq m$, then move \emph{x} to the front, else move \emph{x} to  position \emph{m} in the list.
	
	\begin{figure}[!ht]
		\centering
		\includegraphics[width=7cm,height=7cm]{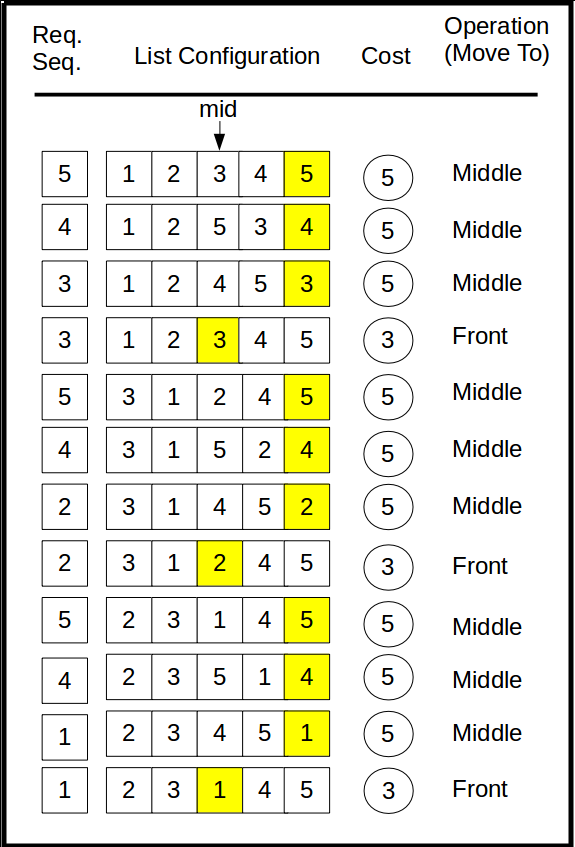}
		\caption{Illustration of MFM}
	\end{figure} 	
	We illustrate \emph{MFM} in \emph{figure-1}. Let $L=$ $<1,2,3,4,5>$ be the initial list configuration and $\sigma=$ $<5, 4, 3, 3, 5, 4, 2, 2, 5, 4, 1, 1>$ be a request sequence. Let \emph{m} be the middle position of the list \emph{L}. When a request is accessed in any position from middle to last position, an access cost is incurred based on the position of the requested item in the list and the item is moved to the middle position with free exchanges. Similarly when requested item is accessed in any of the position from $1$ to $m-1$ in the list, it will incur some access cost and the accessed item is moved to the first position of the list with free exchanges. Here the total cost incurred by the \emph{MFM} algorithm is $54$ as shown in \emph{figure-1}.  
	\section{Our Contribution: Competitive Analysis of Move-To-Front-or-Middle Algorithm }
	In this section, we provide a \emph{2-competitive} bound on the competitive ratio of the $MFM$ algorithm. In our analysis we consider two different optimal offline strategies i.e the $\emph{static opt}$ and the $\emph{dynamic opt}$.  The static optimum offline algorithm is denoted as \emph{STAT} and the dynamic optimum offline algorithm is denoted by \emph{OPT}. it is important to design an optimal offline algorithm that must be giving the tight upper bound on the total cost. For our analysis we consider two different optimal offline strategies i.e the $\emph{static opt}$ and the $\emph{dynamic opt}$. We show that how different optimal strategies can affect the lower bound of $MFM$ algorithm.
	
	\emph{Nature of Static Optimum Offline Algorithm} The static optimum offline algorithm is denoted as \emph{STAT} which initially reorganize the list before serving the request sequence. The reorganization is done on the basis of decreasing frequency of the requested items in $\sigma$ by using only the minimum paid exchanges. Once the list is reordered, it remained same till the complete requested items are served.
	
	\emph{Nature of Dynamic Optimum Offline Algorithm} In contrast to the $\emph{STAT}$, dynamic optimum offline algorithm perform dynamic list rearrangement. Let this optimal strategy be denoted by $\emph{Dyn\_OPT}$. Every time an item is served by $\emph{Dyn\_OPT}$ it is moved to the front of the list. We analyse in detail that how $\emph{Dyn\_OPT}$ helps in providing an upper bound on the total cost and provides a bound for the $\emph{MFM}$ algorithm.
	
	\subsection{Basic Notations and Definition:} Let \emph{L} be the list of \emph{l} distinct items repressed as $<x_{1}, x_{2}, ..., x_{m},..., x_{l}>$ and $\sigma$ be the request sequence of \emph{n} items which can be represented as $\sigma =<\sigma_{1},\sigma_{2},....,\sigma_{n}>$. Let $\sigma_{i}$ represents the $i^{th}$ item in the request sequence $\sigma$ where $\sigma_{i} \in \{x_{1}, x_{2}, ..., x_{n}\}$ and $1 \leq i \leq n$. Let $x_{m}$ be the middle position of the list \emph{L}, where $m=$ $\lceil{\frac{l}{2}}\rceil$. The total access cost incurred by \emph{MFM} to serve $\sigma$ on the given list \emph{L} is denoted by $C_{MFM}(L,\sigma)$.
	
	\subsection{Upper Bound Analysis:}

	\begin{lemma}
		Competitive ratio of \emph{MFM} is not bounded by a factor 2 with respect to \emph{STAT}.
	\end{lemma}
	
	\begin{proof}\emph{Cost of MFM:} Let the initial configuration of the list \emph{L} be $<x_{1}, x_{2},.., x_{m},\\ x_{m+1},..., x_{l}>$ of size \emph{l}, where $x_{i}$ represents the item present at the $x^{th} $ position of the list i.e $i \in \{ 1,2,..m-1, m, m+1..,l \}$ and \emph{m} represents the middle position of the list. The strategy to obtain the cruel request sequence is to access the last item of the list in each request. The cruel request sequence $\sigma$ can be represented as $<x_{l}, x_{l-1},.., x_{m+2}, x_{m+1}, x_{m}, x_{l},\\ x_{l-1}, ....., x_{m+2}, x_{m+1}, x_{m}, x_{l},... >$ of size \emph{n}. Every time an item is accessed, it is moved to the $m^{th}$ position of the list and the rest of the items from position $m$ to position $l-1$ get shifted to one position forward in the list. Hence the cost incurred by \emph{MFM} to serve every item is \emph{l}. So the total access cost incurred by \emph{MFM} to serve $\sigma$ is $C_{MFM}(L,\sigma)$ is $nl$. 
		
		\emph{Cost of STAT:} Since we use the static optimum offline algorithm \emph{STAT} that rearrange the list according to the non increasing order of the requests present in $\sigma$, the pattern of the cruel request sequence affects in obtaining the total cost. The cruel request sequence either can be viewed as the repeated patterns of $l-m+1$ items with the same frequency or different frequency i.e either the size of $\sigma$ be $n=k(l-m+1)$ items or $n\neq k(l-m+1)$, where \emph{k} denotes the number of items are repeated.  
		
		\paragraph{Case:1 $n=k(l-m+1)$} Since all the items are repeated \emph{k} times in the cruel request sequence, $\sigma$ can be represented as $<(x_{l}, x_{l-1},..,x_{m+2}, x_{m+1}, x_{m})^k>$, where the frequency of all the items are same. Before serving the $\sigma$, \emph{STAT} uses the minimum paid exchanges to bring all the $l-m+1$ item to the first \emph{l-m+1} position of the list \emph{L}. Here the cost incurred to perform the minimum number of paid exchanges by \emph{STAT} is $(l-m+1)(m-1)$. The total access cost incurred by \emph{STAT} to serve $\sigma$ is $k(\sum_{i=1}^{l-m+1}(i))$. Hence the total cost incurred by \emph{STAT} to serve $\sigma$ is the sum of paid exchange cost and access cost i.e $C_{STAT}(L,\sigma)$ is $(l-m+1)(m-1) + k(\sum_{i=1}^{l-m+1}(i))$.  
		
		Hence the competitive ratio we obtain for the \emph{MFM} with respect to \emph{STAT} is $\frac{nl}{(l-m+1)(m-1) + k(\sum_{i=1}^{l-m+1}(i))}$. The value of \emph{k} grows with the size of $\sigma$.\\
		
		\begin{tabular}{ |c|c|c| } 
			\hline
			List Size & Value of k  & Competitive Ratio $(\frac{k[l(l-m+1)]}{(l-m+1)(m-1) + k(\sum_{i=1}^{l-m+1}(i))})$  \\
			\hline
			\emph{l=5} & \emph{k=2} & $\frac{30}{18}=1.66$  \\
			\hline
			\emph{l=5} & \emph{k=4} & $\frac{60}{30}=2.00$  \\
			\hline
			\emph{l=5} & \emph{k=8}  & $\frac{120}{54}=2.22$  \\
			\hline
			\emph{l=5} & \emph{k=10}  & $\frac{150}{66}=2.27$ \\
			\hline
			\emph{l=5} & \emph{k=100}  & $\frac{1500}{606}=2.47$ \\
			\hline
			\emph{l=5} & \emph{k=1000}  & $\frac{15000}{6006}=2.49$ \\
			\hline
			\emph{l=5} & \emph{k=10000}  & $\frac{150000}{60006}=2.499$ \\
			\hline
		\end{tabular}\\
		\\
		
\paragraph{ Case:2 $n\neq k(l-m+1)$} In this case, not all the $l-m+1$ items are occurred in $\sigma$ with equal frequency. Since $n \gg l$, there are some items in the $\sigma$ which are repeated with frequency one more than the other items. Here $\sigma$ can be represented as $<(x_{l}, x_{l-1},..,x_{m+2}, x_{m+1}, x_{m})^k,x_{l}, x_{l-1},..,x_{m+2}, x_{m+1}>$. So to serve this cruel request sequence $\sigma$, \emph{STAT} follows three basic steps: in \emph{Step-1}, it brings the block of $l-m+1$ items to the front of the list by giving the minimum paid exchanges, in \emph{Step-2}, it perform the rearrangement of the list according to the non increasing order of the requests and in \emph{Step-3}, it access the list without doing any list rearrangement.
		
		\emph{STAT} incurs a $(l-m+1)(m-1)$ paid exchange cost in \emph{Step-1} while bringing all the block of $l-m+1$ items denoted by $C_{paid\_block}(L,\sigma)=$ $(l-m+1)(m-1)$. The minimum paid exchanges required to sort the items is $C_{paid\_sort}(L,\sigma)=$ $\sum_{i=0}^{[n-k(l-m+1)]-1}(l-m-i)$. In\emph{Step-3}, it incurs an total access cost as $C_{access}(L,\sigma)=$ $k[\sum_{i=1}^{l-m+1}(i)] +\sum_{i=1}^{[n-k(l-m+1)]-1}(i)$.
		
		Hence the total cost incurred by \emph{STAT} to serve $\sigma$ is the sum of paid exchange cost and access cost i.e $C_{STAT}(L,\sigma)$ is $C_{paid\_block}(L,\sigma)+C_{paid\_sort}(L,\sigma)+ C_{access}(L,\sigma)=$ $(l-m+1)(m-1)+ \sum_{i=0}^{[n-k(l-m+1)]-1}(l-m-i)+ k[\sum_{i=1}^{l-m+1}(i)] +\sum_{i=1}^{[n-k(l-m+1)]-1}(i)$.
		
		Hence the competitive ratio we obtain for the \emph{MFM} with respect to \emph{STAT} is $\frac{nl}{(l-m+1)(m-1)+ \sum_{i=0}^{[n-k(l-m+1)]-1}(l-m-i)+ k[\sum_{i=1}^{l-m+1}(i)] +\sum_{i=1}^{[n-k(l-m+1)]-1}(i)}$. The value of \emph{k} grows with the size of $\sigma$.\\
		
		\begin{tabular}{ |c|c|c| } 
			\hline
			List Size & Value of k  & Competitive Ratio $(\frac{nl}{C_{paid\_block}(L,\sigma)+C_{paid\_sort}(L,\sigma)+ C_{access}(L,\sigma)})$  \\
			\hline
			\emph{l=6} & \emph{k=1} & $\frac{42}{30}=1.42$  \\
			\hline
			\emph{l=6} & \emph{k=2} & $\frac{66}{40}=1.65$  \\
			\hline
			\emph{l=6} & \emph{k=3}  & $\frac{90}{50}=1.8$  \\
			\hline
			\emph{l=6} & \emph{k=4}  & $\frac{114}{60}=1.9$ \\
			\hline
			\emph{l=6} & \emph{k=5}  & $\frac{138}{70}=1.97$ \\
			\hline
			\emph{l=6} & \emph{k=6}  & $\frac{162}{80}=2.025$ \\
			\hline
			\emph{l=6} & \emph{k=7}  & $\frac{186}{90}=2.066$ \\
			\hline
			\emph{l=6} & \emph{k=8}  & $\frac{210}{100}=2.10$ \\
			\hline
			\emph{l=6} & \emph{k=9}  & $\frac{234}{110}=2.12$ \\
			\hline
			\emph{l=6} & \emph{k=10}  & $\frac{258}{120}=2.15$ \\
			\hline
			
		\end{tabular}\\
		\\

		\emph{Observation:} \begin{itemize}
			\item As the size of $\sigma$ grows the competitive ratio also grows, then what could be the potential bound on the competitive ratio for larger list as well as request sequence size?
			\item The competitive ratio is not constant due to the nature of static optimum offline algorithm, then can we provide an upper bound on the cost of optimum offline algorithm by introducing dynamic optimum offline algorithm?
		\end{itemize}
	\end{proof}

	\begin{theorem}
		Competitive ratio of \emph{MFM} converges to 4 for large size list.
	\end{theorem}
	
	\begin{proof}
		Here we extend the competitive analysis of \emph{MFM} with respect to \emph{STAT}. We observe that the competitive ratio i.e $\frac{C_{MFM}(L,\sigma)}{C_{STAT}(L,\sigma)} > 2$. We follow the same method of proof for \emph{MFM} which we have used for \emph{MFLP}, and investigate the converging point. When the list \emph{L} tends to infinity, then the competitive ratio slowly converges to 4. Because when $\lim_{l\to\infty}({\frac{l(l-m+1)}{\sum_{i=1}^{l-m+1}(i)}})=$ $\lim_{l\to\infty}({\frac{l(l-\frac{l}{2}+1)}{\sum_{i=1}^{l-\frac{l}{2}+1}(i)}})=$\\ $\lim_{l\to\infty}({\frac{2l(l-\frac{l}{2}+1)}{(l-\frac{l}{2}+1)(l-\frac{l}{2}+2)}})=$ $\lim_{l\to\infty}({\frac{2l}{l-\frac{l}{2}+2}})=$ $\lim_{l\to\infty}({\frac{4l}{2l-l+4}})$ $\approx 4$.
		
		\emph{Observation:} \begin{itemize}
			\item An online deterministic list update algorithm that involves $(l-m+1)$ number of elements while constructing the worst request sequence, is not a $2-competitive$ algorithm with respect to static optimum offline algorithm.
			\end{itemize}
	\end{proof}
	
\begin{theorem}
A large class of deterministic list update algorithms are not $2-competitive$ with respect to static optimum offline algorithm, when the online algorithms move the recently accessed item to the position \emph{x} when the position of the accessed item is greater than the position \emph{x}, else move it to the front of the list, where $x$ varies from the position $logL$ to $m$. 
\end{theorem}

\begin{proof}
Our proof is followed by two interesting observations that when an recently accessed item is move to the front of the list in \emph{MTF}, then it is \emph{2-competitive} with respect to the static optimum offline algorithm $STAT$. But when an accessed item is moved to the middle position in \emph{MFM} it is not $2-competitive$. In one of our recent paper \cite{Baisakh}, we have raised an open question that for what value of \emph{p}, an online algorithm is \emph{2-competitive} where the value of \emph{p} may vary from \emph{2 to m}? In that paper, we have shown that when the value of \emph{p} is $log_{2}L$, then it is not \emph{2-competitive} with respect to $STAT$ for smaller size list. In \emph{MFLP}  when an item \emph{x} is accessed, move \emph{x} to the $p^{th}$ position of the list \emph{L}, where $p= \lceil log_{2}L \rceil$ if the position of the \emph{x} is greater then \emph{p}, else move \emph{x} to the front of the list. In this paper we have shown that $MFM$ is not $2-competitive$ with respect to $STAT$. So from our investigation we can observe that when the recently item is moved to the position \emph{x}, where \emph{x} can be varied from the position $logL$ to $m$, then the nature of such algorithms are \emph{2-competitive}. But our theoretical study also shown an interesting observation that, for smaller size list these algorithms are not \emph{2-competitive} with respect to \emph{STAT}. But for larger size list the competitive ratio may vary from \emph{ 2 to 4}. Because the competitive ratio of \emph{MFLP} converges to 2 for a very large size list, whereas \emph{MFM} converges to \emph{4}, with respect to the \emph{STAT}.
\end{proof}

\subsection{Lower Bound Analysis:}

To obtain the lower bound on the competitive ratio of \emph{MFM}, we have to obtain the upper bound cost of optimum offline algorithm. So here we have presented a dynamic optimum offline algorithm which gives more cost as compare to the static optimum offline algorithm on the given cruel request sequence.

\begin{theorem}
	\emph{MFM} is \emph{2-competitive} wit respect to dynamic optimum offline algorithm.
\end{theorem}

\begin{proof}
	\emph{\bf Cost of MFM:} 
	We consider the same cruel request sequence for which $MFM$ algorithm incurs the total worst case cost as $k(l(l-m+1))$.\\
	
	\emph{\bf Cost of Dynamic Optimum Offline Algorithm (OPT):} Let the dynamic optimum offline algorithm is denoted by \emph{OPT}. Every time the $OPT$ access an item, it brings it to the front of the list. So to access the first $(l-m+1)$ items the algorithm incurs the cost as $l(l-m+1)$. In the rest of requested items i.e $(n-(l-m+1))$, items are repeated $k-1$ times. Hence to serve the rest of items in the cruel request sequence $OPT$ incurs $(k-1)[(l-m+1)(l-m+1)]$. So the total cost incurred by the dynamic offline algorithm $OPT$ i.e $C_{OPT}(L,\sigma)$ is $l(l-m+1) + (k-1)[(l-m+1)(l-m+1)]$. Hence the competitive ratio we obtain for the \emph{MFM} with respect to \emph{dynamic optimum offline algorithm (OPT)} is $\frac{C_{MFM}(L,\sigma)}{C_{OPT}(L,\sigma)}=$ $\frac{k(l(l-m+1))}{l(l-m+1) + (k-1)[(l-m+1)(l-m+1)]}=$ $\frac{l(l-m+1)+(k-1)(l(l-m+1))}{l(l-m+1) + (k-1)[(l-m+1)(l-m+1)]}$.\\

	\begin{tabular}{ |c|c|c| } 
		\hline
		List Size & Value of k  & Competitive Ratio $(\frac{l(l-m+1)+(k-1)(l(l-m+1))}{l(l-m+1) + (k-1)[(l-m+1)(l-m+1)]})$  \\
		\hline
		\emph{l=5} & \emph{k=2} & $\frac{30}{24}=1.25$  \\
		\hline
		\emph{l=5} & \emph{k=4} & $\frac{60}{42}=1.42$  \\
		\hline
		\emph{l=5} & \emph{k=8}  & $\frac{120}{78}=1.53$  \\
		\hline
		\emph{l=5} & \emph{k=10}  & $\frac{150}{96}=1.562$ \\
		\hline
		\emph{l=5} & \emph{k=100}  & $\frac{1500}{906}=1.655$ \\
		\hline
		\emph{l=5} & \emph{k=1000}  & $\frac{15000}{9006}=1.665$ \\
		\hline
		\emph{l=5} & \emph{k=10000}  & $\frac{150000}{90006}=1.666$ \\
		\hline
	\end{tabular}\\
	
	\emph{Observation:}\begin{itemize}
		\item The new dynamic optimum offline algorithm provides an upper bound to the total cost incurred by \emph{OPT}.
		\item The upper bound in the competitive ratio ensures to achieve a lower bound on the \emph{MFM}. 
	\end{itemize}
\end{proof}
\section{Conclusion}
 In this paper, we analysed the \emph{Move-to-Front-or-Middle (MFM)} algorithm using competitive analysis by addressing one of an open question raised by \emph{Albers}. \emph{Move-To-Front-or-Middle (MFM)} was experimentally studied and observed to be performing better than \emph{MTF} algorithm by using the Calgary Corpus and Canterbury Corpus data set.  We made a first attempt to find the competitiveness of \emph{MFM} algorithm. Our new results showed that \emph{MFM} is not \emph{2-competitive} with respect to static optimum offline algorithm, whereas it is \emph{2-competitive} with respect to dynamic optimum offline algorithm. Our new theoretical results may open up a new direction of research in the online list update problem by characterising the structure of competitive and non competitive deterministic online algorithms.

\end{document}